%% file: paper-graphnet.tex
\documentclass{article}
\usepackage[linesnumbered,ruled]{algorithm2e}
\usepackage[final,nonatbib]{nips_2017}
\usepackage[utf8]{inputenc} 
\usepackage[T1]{fontenc}    
\usepackage{hyperref}       
\usepackage{booktabs}       
\usepackage{amsfonts}       
\usepackage{microtype}      
\usepackage{xcolor}
\usepackage{url}
\usepackage{verbatim} 
\usepackage{graphicx}
\usepackage{caption} 
\usepackage{subcaption} 
\usepackage{multirow}
\usepackage{xspace}
\usepackage{epsfig}
\usepackage{amsmath}
\usepackage{amsthm}
\usepackage{amssymb}
\usepackage{times}
\usepackage{xr}

\newcommand{\xhdr}[1]{{\noindent\bfseries #1}.}

\newcommand{\name}{GraphSAGE}

\newcommand{\V}{\mathcal{V}}
\newcommand{\E}{\mathcal{E}}
\newcommand{\G}{\mathcal{G}}

\newcommand{\mb}{\mathbf}
\newcommand{\cut}[1]{}

\newtheorem{theorem}{Theorem}
\newtheorem{corollary}[theorem]{Corollary}
\newtheorem{lemma}{Lemma}

\graphicspath{{./FIG/}}

\title{Inductive Representation Learning on Large Graphs}

\def\sharedaffiliation{%
\end{tabular}
\begin{tabular}{c}}
\newcommand*\samethanks[1][\value{footnote}]{\footnotemark[#1]}

\author{
William L. Hamilton\thanks{The two first authors made equal contributions.}\\
\texttt{wleif@stanford.edu}
\And 
Rex Ying\samethanks\\
\texttt{rexying@stanford.edu}
\And 
Jure Leskovec \\
\texttt{jure@cs.stanford.edu}
\vspace{10pt}
\sharedaffiliation
Department of Computer Science\\
Stanford University\\
Stanford, CA, 94305
}

\begin{document}
\maketitle
\setcounter{footnote}{0}

\begin{abstract}
\input{000abstract}
\end{abstract}

\input{010intro}

\input{020related}

\input{030proposed}

\input{040experiments}

\input{060conclusion}

\vspace{10pt}
\subsection*{Acknowledgments}
\vspace{-5pt}
The authors thank Austin Benson, Aditya Grover, Bryan He, Dan Jurafsky, Alex Ratner, Marinka Zitnik, and Daniel Selsam for their helpful discussions and comments on early drafts. 
The authors would also like to thank Ben Johnson for his many useful questions and comments on our code and Nikhil Mehta and Yuhui Ding for catching some minor errors in a previous version of the appendix. 
This research has been supported in part by NSF IIS-1149837, DARPA SIMPLEX,
Stanford Data Science Initiative, Huawei, and Chan Zuckerberg Biohub.
WLH was also supported by the SAP Stanford Graduate Fellowship and an NSERC PGS-D grant. 
The views and conclusions expressed in this material are those of the authors
and should not be interpreted as necessarily representing the official policies or endorsements, either expressed or
implied, of the above funding agencies, corporations, or the U.S. and Canadian governments. 
\newpage
\bibliographystyle{abbrv}
\bibliography{refs}
\newpage
\section*{{\Large Appendices}}
\appendix
\input{070appendix}

\end{document}

%% file: 000abstract.tex
Low-dimensional embeddings of nodes in large graphs have proved extremely useful in a variety of
prediction tasks, from content recommendation to identifying protein functions. However,
most existing approaches require that all nodes in the graph are present during training of the embeddings; these previous approaches are inherently
{\em transductive} and do not naturally generalize to unseen nodes. 
Here we present \name, a general {\em inductive} framework that leverages node feature information (e.g., text attributes) to efficiently generate node embeddings for previously unseen data.
Instead of training individual embeddings for each node, we learn a function that generates embeddings by sampling and aggregating features from a node's local neighborhood.
Our algorithm outperforms strong baselines on three inductive node-classification benchmarks: we classify the category of unseen nodes in evolving information graphs based on citation and Reddit post data, \cut{in both supervised and unsupervised settings} and we show that our algorithm generalizes to completely unseen graphs using a multi-graph dataset of protein-protein interactions. 

%% file: 010intro.tex

\section{Introduction}
\label{sec:intro}

Low-dimensional  vector embeddings of nodes in large graphs\footnote{While it is common to refer to these data structures as social or biological \emph{networks}, we use the term \emph{graph} to avoid ambiguity with neural network terminology.} have proved extremely useful as feature inputs for a wide variety of prediction and graph analysis tasks \cite{cao2015grarep,grover2016node2vec,perozzi2014deepwalk,tang2015line,wang2016structural}.
The basic idea behind node embedding approaches is to use dimensionality reduction techniques to distill the high-dimensional information about a node's graph neighborhood into a dense vector embedding. 
These node embeddings can then be fed to downstream machine learning systems and aid in tasks such as node classification, clustering, and link prediction \cite{grover2016node2vec,perozzi2014deepwalk,tang2015line}.

However, previous works have focused on embedding nodes from a single fixed graph, and many real-world applications require embeddings to be quickly generated for unseen nodes, or entirely new (sub)graphs. 
This inductive capability is essential for high-throughput, production machine learning systems, which operate on evolving graphs and constantly encounter unseen nodes (e.g., posts on Reddit, users and videos on Youtube).
An inductive approach to generating node embeddings also facilitates generalization across graphs with the same form of features:
for example, one could train an embedding generator on protein-protein interaction graphs derived from a model organism, and then easily produce node embeddings for data collected on new organisms using the trained model. 

The inductive node embedding problem is especially difficult, compared to the transductive setting, because generalizing to unseen nodes  requires ``aligning'' newly observed subgraphs to the node embeddings that the algorithm has already optimized on. 
An inductive framework must learn to recognize structural properties of a node's neighborhood that reveal both the node's local role in the graph, as well as its global position. 
 
 Most existing approaches to generating node embeddings are inherently transductive.
 The majority of these approaches directly optimize the embeddings for each node using matrix-factorization-based objectives, and do not naturally generalize to unseen data, since they make predictions on nodes in a single, fixed graph \cite{cao2015grarep,grover2016node2vec,ng2001spectral,perozzi2014deepwalk,tang2015line,
 wang2016structural,wang2017community,xu2017embedding}.
 These approaches can be modified to operate in an inductive setting (e.g., \cite{perozzi2014deepwalk}), but these modifications tend to be computationally expensive, requiring additional rounds of gradient descent before new predictions can be made. 
 There are also recent approaches to learning over graph structures using convolution operators that offer promise as an embedding methodology \cite{kipf2016semi}.
 So far, graph convolutional networks (GCNs)  have only been applied in the transductive setting with fixed graphs \cite{kipf2016semi,kipf2016variational}. 
 In this work we both extend GCNs to the task of inductive unsupervised learning and propose a framework that generalizes the GCN approach to use trainable aggregation functions (beyond simple convolutions). 

\xhdr{Present work} We propose a general framework, called \name\ (\textsc{sa}mple and aggre\textsc{g}at\textsc{e}), for inductive node embedding.
Unlike embedding approaches that are based on matrix factorization, we leverage node features (e.g., text attributes, node profile information, node degrees) in order to learn an embedding function that generalizes to unseen nodes. 
By incorporating node features in the learning algorithm, we simultaneously learn the topological structure of each node's neighborhood as well as the distribution of node features in the neighborhood.
While we focus on feature-rich graphs (e.g., citation data with text attributes, biological data with functional/molecular markers), our approach can also make use of structural features that are present in all graphs (e.g., node degrees).
Thus, our algorithm can also be applied to graphs without node features.

Instead of training a distinct embedding vector for each node, we train a set of {\em aggregator functions} that learn to aggregate feature information from a node's local neighborhood (Figure \ref{fig:overview}).
Each aggregator function aggregates information from a different number of hops, or search depth, away from a given node. 
At test, or inference time, we use our trained system to generate embeddings for entirely unseen nodes by applying the learned aggregation functions. 
Following previous work on generating node embeddings, we design an unsupervised loss function that allows \name\ to be trained without task-specific supervision.
We also show that \name\ can be trained in a fully supervised manner. 

\begin{figure}
\vspace{-20pt}
\centering
\includegraphics[width=0.99\textwidth]{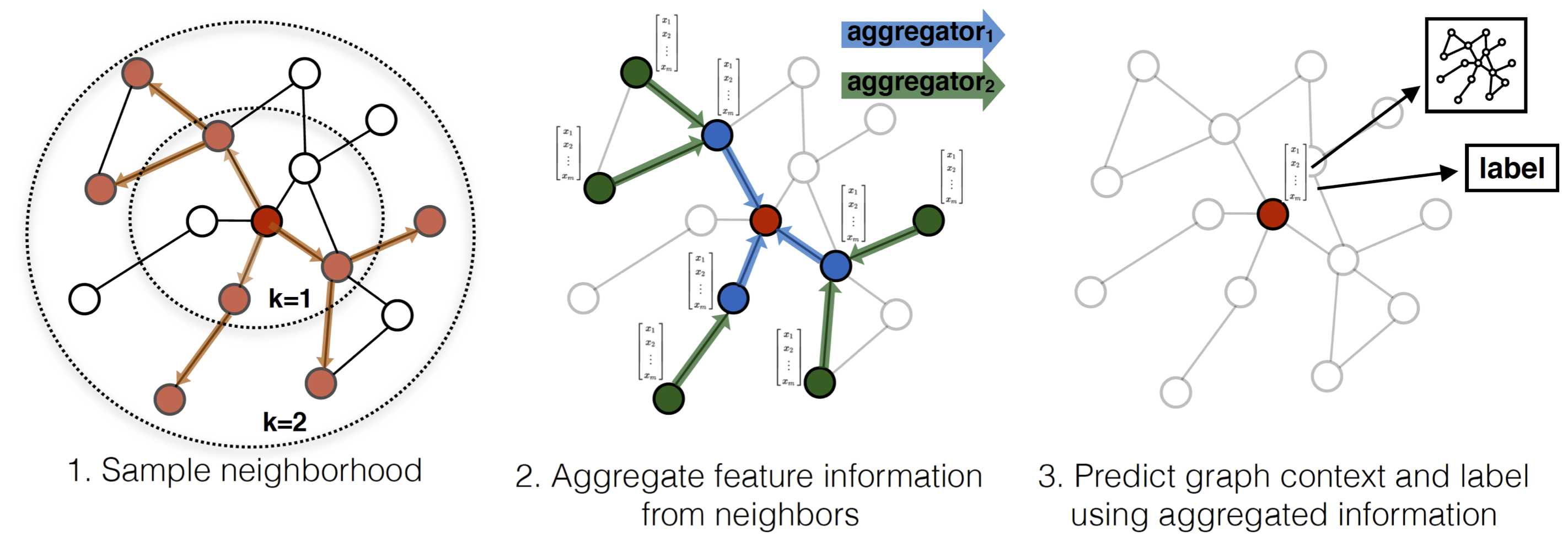}
\caption{Visual illustration of the \name\ sample and aggregate approach.}
\vspace{-15pt}
\label{fig:overview}
\end{figure}

We evaluate our algorithm on three node-classification benchmarks, which test \name's ability to generate useful embeddings on unseen data. 
We use two evolving document graphs based on citation data and Reddit post data (predicting paper and post categories, respectively), and a multi-graph generalization experiment based on a dataset of protein-protein interactions (predicting protein functions). 
  Using these benchmarks, we show that our approach is able to effectively generate representations for unseen nodes and outperform relevant baselines by a significant margin: across domains, our supervised approach improves classification F1-scores by an average of 51\% compared to using node features alone and \name\ consistently outperforms a strong, transductive baseline \cite{perozzi2014deepwalk}, despite this baseline taking ${\sim}100\times$ longer to run on unseen nodes.
We also show that the new aggregator architectures we propose provide significant gains (7.4\% on average) compared to an aggregator inspired by graph convolutional networks \cite{kipf2016semi}.
 Lastly, we probe the expressive capability of our approach and show, through theoretical analysis, that \name\ is capable of learning structural information about a node's role in a graph, despite the fact that it is inherently based on features (Section \ref{sec:theory}). 

%% file: 020related.tex

\section{Related work}

Our algorithm is conceptually related to previous node embedding approaches, general supervised approaches to learning over graphs, and recent advancements in applying convolutional neural networks to graph-structured data.\footnote{In the time between this papers original submission to NIPS 2017 and the submission of the final, accepted (i.e., ``camera-ready'') version, there have been a number of closely related (e.g., follow-up) works published on pre-print servers. For temporal clarity, we do not review or compare against these papers in detail.}

\xhdr{Factorization-based embedding approaches}
There are a number of recent node embedding approaches that learn low-dimensional embeddings using random walk statistics and matrix factorization-based learning objectives \cite{cao2015grarep,grover2016node2vec,perozzi2014deepwalk,tang2015line,wang2016structural}.
These methods also bear close relationships to more classic approaches to spectral clustering \cite{ng2001spectral}, multi-dimensional scaling \cite{kruskal1964multidimensional}, as well as the PageRank algorithm \cite{page1999pagerank}. 
Since these embedding algorithms directly train node embeddings for individual nodes, they are inherently transductive and, at the very least, require expensive additional training (e.g., via stochastic gradient descent) to make predictions on new nodes.
In addition, for many of these approaches (e.g., \cite{grover2016node2vec,perozzi2014deepwalk,tang2015line,wang2016structural}) the objective function is invariant to orthogonal transformations of the embeddings, which means that the embedding space does not naturally generalize between graphs and can drift during re-training.
One notable exception to this trend is the Planetoid-I algorithm introduced by Yang et al.\ \cite{yang2016revisiting}, which is an inductive, embedding-based approach to semi-supervised learning.
However, Planetoid-I does not use any graph structural information during inference; instead, it uses the graph structure as a form of regularization during training.
Unlike these previous approaches, we leverage feature information in order to train a model to produce embeddings for unseen nodes. 

\xhdr{Supervised learning over graphs}
Beyond node embedding approaches, there is a rich literature on supervised learning over graph-structured data. 
This includes a wide variety of kernel-based approaches, where feature vectors for graphs are derived from various graph kernels (see \cite{shervashidze2011weisfeiler} and references therein). 
There are also a number of recent neural network approaches to supervised learning over graph structures \cite{dai2016discriminative,gori2005new,li2015gated,scarselli2009graph}.
Our approach is conceptually inspired by a number of these algorithms.
However, whereas these previous approaches attempt to classify entire graphs (or subgraphs), the focus of this work is generating useful representations for individual nodes. 

\xhdr{Graph convolutional networks}
In recent years, several convolutional neural network architectures for learning over graphs have been proposed (e.g., \cite{bruna2013spectral,duvenaud2015convolutional,
defferrard2016convolutional,kipf2016semi,niepert2016learning}).
The majority of these methods do not scale to large graphs or are designed for whole-graph classification (or both) \cite{bruna2013spectral,duvenaud2015convolutional,
defferrard2016convolutional,niepert2016learning}.
However, our approach is closely related to the graph convolutional network (GCN), introduced by Kipf et al. \cite{kipf2016semi,kipf2016variational}. 
The original GCN algorithm \cite{kipf2016semi} is designed for semi-supervised learning in a transductive setting, and the exact algorithm requires that the full graph Laplacian is known during training. 
A simple variant of our algorithm can be viewed as an extension of the GCN framework to the inductive setting, a point which we revisit in Section \ref{sec:aggs}.

%% file: 030proposed.tex

\section{Proposed method: \name}

The key idea behind our approach is that we learn how to aggregate feature information from a node's local neighborhood (e.g., the degrees or text attributes of nearby nodes). 
We first describe the  \name\ embedding generation (i.e., forward propagation) algorithm, which generates embeddings for nodes assuming that the \name\ model parameters are already learned (Section \ref{sec:foward}). 
We then describe how the \name\ model parameters can be learned using standard stochastic gradient descent and backpropagation techniques (Section \ref{sec:learning}). 

%

\subsection{Embedding generation (i.e., forward propagation) algorithm}\label{sec:foward} 

\label{sec:method}
\begin{algorithm}
\caption{\name\ embedding generation (i.e., forward propagation) algorithm}
\label{alg:basic}
	\SetKwInOut{Input}{Input}\SetKwInOut{Output}{Output}
    \Input{~Graph $\G(\V,\E)$; input features $\{\mb{x}_v, \forall v\in \V\}$; depth $K$; weight matrices $\mb{W}^{k}, \forall k \in \{1,...,K\}$; non-linearity $\sigma$; differentiable aggregator functions $\textsc{aggregate}_k, \forall k \in \{1,...,K\}$; neighborhood function $\mathcal{N} : v \rightarrow 2^{\V}$}
    \Output{~Vector representations $\mb{z}_v$ for all $v \in \V$}
    \BlankLine
    $\mb{h}^0_v \leftarrow \mb{x}_v, \forall v \in \V$ \;
    \For{$k=1...K$}{
    	  \For{$v \in \V$}{
    	  $\mb{h}^{k}_{\mathcal{N}(v)} \leftarrow \textsc{aggregate}_k(\{\mb{h}_u^{k-1}, \forall u \in \mathcal{N}(v)\})$\;
    	  		$\mb{h}^k_v \leftarrow \sigma\left(\mb{W}^{k}\cdot\textsc{concat}(\mb{h}_v^{k-1}, \mb{h}^{k}_{\mathcal{N}(v)})\right)$
    	  }
    	  $\mb{h}^{k}_v\leftarrow \mb{h}^{k}_v/ \|\mb{h}^{k}_v\|_2, \forall v \in \V$
    	}
     $\mb{z}_v\leftarrow \mb{h}^{K}_v, \forall v \in \V$ 
\end{algorithm}

In this section, we describe the embedding generation, or forward propagation algorithm (Algorithm \ref{alg:basic}), which assumes that the model has already been trained and that the parameters are fixed. 
In particular, we assume that we have learned the parameters of $K$ aggregator functions (denoted $\textsc{aggregate}_k, \forall k \in \{1,...,K\}$), which aggregate information from node neighbors, as well as a set of weight matrices $\mb{W}^{k}, \forall k \in \{1,...,K\}$, which are used to propagate information between different layers of the model or ``search depths''. 
Section \ref{sec:learning} describes how we train these parameters. 

The intuition behind Algorithm \ref{alg:basic} is that at each iteration, or search depth, nodes aggregate information from their local neighbors, and as this process iterates, nodes incrementally gain more and more information from further reaches of the graph.

Algorithm \ref{alg:basic} describes the embedding generation process in the case where the entire graph, $\G=(\V, \E)$,  and features  for all nodes $\mb{x}_v, \forall v \in \V$, are provided as input.
We describe how to generalize this to the minibatch setting below. 
Each step in the outer loop of Algorithm \ref{alg:basic} proceeds as follows, where $k$ denotes the current step in the outer loop (or the depth of the search) and $\mb{h}^{k}$ denotes a node's representation at this step: First, each node $v\in\V$ aggregates the representations of the nodes in its immediate neighborhood, $\{\mb{h}^{k-1}_u, \forall u \in \mathcal{N}(v)\}$, into a single vector $\mb{h}^{k-1}_{\mathcal{N}(v)}$. 
Note that this aggregation step depends on the representations generated at the previous iteration of the outer loop (i.e., $k-1$),  and the $k=0$ (``base case'') representations are defined as the input node features. 
After aggregating the neighboring feature vectors, \name\ then concatenates the node's current representation, $\mb{h}_v^{k-1}$, with the aggregated neighborhood vector, $\mb{h}^{k-1}_{\mathcal{N}(v)}$, and this concatenated vector is fed through a fully connected layer with nonlinear activation function $\sigma$, which transforms the representations to be used at the next step of the algorithm (i.e.,  $ \mb{h}_v^{k}, \forall v \in \mathcal{V}$). 
For notational convenience, we denote the final representations output at depth $K$ as $\mb{z}_v \equiv \mb{h}^K_v , \forall v \in \mathcal{V}$.
The aggregation of the neighbor representations can be done by a variety of aggregator architectures (denoted by the \textsc{aggregate} placeholder in Algorithm 1), and we discuss different architecture choices in Section \ref{sec:aggs} below.
\cut{Figure \ref{fig:struct_learn} provides a simple visual illustration of  Algorithm \ref{alg:basic}.\cut{, when we use a hash function as an aggregator and use node degrees as features.}}

To extend Algorithm \ref{alg:basic} to the minibatch setting, given a set of input nodes, we first forward sample the required neighborhood sets (up to depth $K$) and then we run the inner loop (line 3 in Algorithm \ref{alg:basic}), but instead of iterating over all nodes, we compute only the representations that are necessary to satisfy the recursion at each depth (Appendix \ref{sec:minibatch} contains complete minibatch pseudocode).

\xhdr{Relation to the Weisfeiler-Lehman Isomorphism Test}\label{sec:wl}
The \name\ algorithm is conceptually inspired by a classic algorithm for testing graph isomorphism. 
If, in Algorithm \ref{alg:basic}, we (i) set $K=|\V|$, (ii) set the weight matrices as the identity, and (iii) use an appropriate hash function as an aggregator (with no non-linearity), then Algorithm \ref{alg:basic} is an instance of the Weisfeiler-Lehman (WL) isomorphism test, also known as ``naive vertex refinement'' \cite{shervashidze2011weisfeiler}.
If the set of representations $\{\mb{z}_v, \forall v \in \V\}$ output by Algorithm \ref{alg:basic} for two subgraphs are identical then the WL test declares the two subgraphs to be isomorphic.
This test is known to fail in some cases, but is valid for a broad class of graphs \cite{shervashidze2011weisfeiler}.
\name\ is a continuous approximation to the WL test, where we replace the hash function with trainable neural network aggregators. 
Of course, we use \name\ to generate useful node representations--not to test graph isomorphism. 
Nevertheless, the connection between \name\ and the classic WL test provides theoretical context for our algorithm design to learn the topological structure of node neighborhoods. 
\cut{
\begin{figure}
\centering
\includegraphics[width=0.95\textwidth]{FIG/struct_learn.png}
\caption{Two iterations of Algorithm 1 on two similar graphs for a toy example where only degrees are used as features (denoted by different colors) and where we use a simple set hash function as an aggregator. The circled, middle node plays a similar ``bridging'' role in both graphs. After $k=1$ iterations the representation for this central node is not unique (it is the same as the starred nodes in top graph).  After $k=2$ iterations the node has the same unique representation in both graphs.
\cut{Note that if we were to run another iteration of Algorithm \ref{alg:basic}, then the representation of the circled nodes would again become distinct; this highlights, how different depths contain different types of structural information and motivates the concatenation of depth information in Layer 1.}}
\label{fig:struct_learn}
\vspace{-10pt}
\end{figure}
}

\xhdr{Neighborhood definition}
In this work, we uniformly sample a fixed-size set of neighbors, instead of using full neighborhood sets in Algorithm \ref{alg:basic}, in order to keep the computational footprint of each batch fixed.\footnote{Exploring non-uniform samplers is an important direction for future work.}
That is, using overloaded notation, we define $\mathcal{N}(v)$ as a fixed-size, uniform draw from the set $\{u \in \V : (u,v) \in \E\}$, and we draw different uniform samples at each iteration, $k$, in Algorithm \ref{alg:basic}. 
Without this sampling the memory and expected runtime of a single batch is unpredictable and in the worst case $O(|\V|)$.\cut{,  due to the presence of extremely high-degree ``hub''-nodes, which are common in real-world graphs \cite{clauset2009power}. }
In contrast, the per-batch space and time complexity for \name\ is fixed at $O(\prod_{i=1}^{K}S_i)$, where $S_i, i \in \{1,...,K\}$ and $K$ are user-specified constants.
Practically speaking we found that our approach could achieve high performance with $K=2$ and $S_1\cdot S_2 \leq 500$ (see Section~\ref{sec:sensitivity} for details).

\subsection{Learning the parameters of \name}\label{sec:learning}  
In order to learn useful, predictive representations in a fully unsupervised setting, we apply a graph-based loss function to the output representations, $\mb{z}_u, \forall u \in {\V}$, and tune the weight matrices, $\mb{W}^k, \forall k \in \{1,...,K\}$,  and parameters of the aggregator functions via stochastic gradient descent. 
The graph-based loss function encourages nearby nodes to have similar representations, while enforcing that the representations of disparate nodes are highly distinct:
\begin{equation}\label{eq:neg}
J_{\G}(\mb{z}_u) = -\log\left(\sigma(\mb{z}^\top_u\mb{z}_{v}) \right) - Q\cdot\mathbb{E}_{v_{n} \sim P_n(v)}\log\left(\sigma(-\mb{z}^\top_u\mb{z}_{v_{n}})\right),
\end{equation}
where $v$ is a node that co-occurs near $u$ on fixed-length random walk, $\sigma$ is the sigmoid function, $P_n$ is a negative sampling distribution, and $Q$ defines the number of negative samples. 
Importantly, unlike previous embedding approaches, the representations $\mb{z}_u$ that we feed into this loss function are generated from the features contained within a node's local neighborhood, rather than training a unique embedding for each node (via an embedding look-up). 

This unsupervised setting emulates situations where node features are provided to downstream machine learning applications, as a service or in a static repository. 
In cases where representations are to be used only on a specific downstream task, the unsupervised loss (Equation \ref{eq:neg}) can simply be replaced, or augmented, by a task-specific objective (e.g., cross-entropy loss).

\subsection{Aggregator Architectures}\label{sec:aggs}

Unlike machine learning over N-D lattices (e.g., sentences, images, or 3-D volumes), a node's neighbors have no natural ordering; thus, the aggregator functions in Algorithm \ref{alg:basic} must operate over an unordered set of vectors.
Ideally, an aggregator function would be symmetric (i.e., invariant to permutations of its inputs) while still being trainable and maintaining high representational capacity. 
The symmetry property of the aggregation function ensures that our neural network model can be trained and applied to arbitrarily ordered node neighborhood feature sets.
We examined three candidate aggregator functions:

\xhdr{Mean aggregator} Our first candidate aggregator function is the mean operator, where we simply take the elementwise mean of the vectors in $\{\mb{h}_u^{k-1}, \forall u \in \mathcal{N}(v)\}$.
The mean aggregator is nearly equivalent to the convolutional propagation rule used in the transductive GCN framework \cite{kipf2016semi}.
In particular, we can derive an inductive variant of the GCN approach by replacing lines 4 and 5 in Algorithm \ref{alg:basic} with the following:\footnote{Note that this differs from Kipf et al's exact equation by a minor normalization constant \cite{kipf2016semi}.}
\begin{equation}
\mb{h}^k_v \leftarrow \sigma(\mb{W}\cdot\textsc{mean}(\{\mb{h}^{k-1}_v\} \cup \{\mb{h}_u^{k-1}, \forall u \in {\mathcal{N}(v)}\}).
\end{equation}
We call this modified mean-based aggregator {\em convolutional} since it is a rough, linear approximation of a localized spectral convolution \cite{kipf2016semi}.
An important distinction between this convolutional aggregator and our other proposed aggregators is that it does not perform the concatenation operation in line 5 of Algorithm \ref{alg:basic}---i.e., the convolutional aggregator does concatenate the node's previous layer representation $\mb{h}^{k-1}_v$ with the aggregated neighborhood vector $\mb{h}^{k}_{\mathcal{N}(v)}$. 
This concatenation can be viewed as a simple form of a ``skip connection'' \cite{he2016identity} between the different ``search depths'', or ``layers'' of the GraphSAGE algorithm, and it leads to significant gains in performance (Section \ref{sec:exp}).  

\xhdr{LSTM aggregator} We also examined a more complex aggregator based on an LSTM architecture \cite{hochreiter1997long}. 
Compared to the mean aggregator, LSTMs have the advantage of larger expressive capability. However, it is important to note that LSTMs are not inherently symmetric (i.e., they are not permutation invariant), since they process their inputs in a sequential manner.
We adapt LSTMs to operate on an unordered set by simply applying the LSTMs to a random permutation of the node's neighbors. 

\xhdr{Pooling aggregator} The final aggregator we examine is both symmetric and trainable.
In this {\em pooling} approach, each neighbor's vector is independently fed through a fully-connected neural network; following this transformation, an elementwise max-pooling operation is applied to aggregate information across the neighbor set:
\begin{equation}\label{eq:pool}
\textsc{aggregate}_k^{\textrm{pool}} = \max(\{\sigma\left(\mb{W}_{\textrm{pool}}\mb{h}^k_{u_i} + \mb{b}\right), \forall u_i \in \mathcal{N}(v)\}),
\end{equation}
where $\max$ denotes the element-wise max operator and $\sigma$ is a nonlinear activation function. 
In principle, the function applied before the max pooling can be an arbitrarily deep multi-layer perceptron, but we focus on simple single-layer architectures in this work. 
This approach is inspired by recent advancements in applying neural network architectures to learn over general point sets \cite{qi2016pointnet}.
Intuitively, the multi-layer perceptron can be thought of as a set of functions that compute features for each of the node representations in the neighbor set. By applying the max-pooling operator to each of the computed features, the model effectively captures different aspects of the neighborhood set. 
Note also that, in principle, any symmetric vector function could be used in place of the $\max$ operator (e.g., an element-wise mean). 
We found no significant difference between max- and mean-pooling in developments test and thus focused on max-pooling for the rest of our experiments.

%% file: 040experiments.tex

\section{Experiments}
\label{sec:exp}

We test the performance of \name\ on three benchmark tasks: (i) classifying academic papers into different subjects using the Web of Science
citation dataset, (ii) classifying Reddit posts as belonging to different communities, and (iii) classifying protein functions across various biological protein-protein interaction (PPI) graphs. Sections \ref{sec:evolving} and \ref{sec:ppi} summarize the datasets, and the supplementary material contains additional information.
In all these experiments, we perform predictions on nodes that are not seen during training, and, in the case of the PPI dataset, we test on entirely unseen graphs. 

\xhdr{Experimental set-up}
To contextualize the empirical results on our inductive benchmarks, we compare against four baselines: 
a random classifer, a logistic regression feature-based classifier (that ignores graph structure), the DeepWalk algorithm \cite{perozzi2014deepwalk} as a representative factorization-based approach, and a
concatenation of the raw features and DeepWalk embeddings.
We also compare four variants of \name\ that use the different aggregator functions (Section \ref{sec:aggs}).
Since, the ``convolutional'' variant of \name\ is an extended, inductive version of Kipf et al's semi-supervised GCN \cite{kipf2016semi}, we term this variant \name-GCN.
We test unsupervised variants of \name\, trained according to the loss in Equation \eqref{eq:neg}, as well as supervised variants that are trained directly on classification cross-entropy loss.
For all the \name\ variants we used rectified linear units as the non-linearity and set $K=2$ with neighborhood sample sizes $S_1=25$ and $S_2=10$ (see Section \ref{sec:sensitivity} for sensitivity analyses).

For the Reddit and citation datasets, we use ``online'' training for DeepWalk as described in Perozzi et al.~\cite{perozzi2014deepwalk}, where we run a new round of SGD optimization to embed the new test nodes before making predictions (see the Appendix for details).
In the multi-graph setting, we cannot apply DeepWalk, since the embedding spaces generated by running the DeepWalk algorithm on different disjoint graphs can be arbitrarily rotated with respect to each other (Appendix \ref{sec:rotational}).

All models were implemented in TensorFlow \cite{abadi2016tensorflow} with the Adam optimizer \cite{kingma2014adam} (except DeepWalk, which performed better with the vanilla gradient descent optimizer). 
We designed our experiments with the goals of (i) verifying the improvement of \name\ over the baseline approaches (i.e., raw features and DeepWalk) and (ii) providing a rigorous comparison of the different \name\ aggregator architectures. 
In order to provide a fair comparison, all models share an identical implementation of their minibatch iterators, loss function and neighborhood sampler (when applicable). 
Moreover, in order to guard against unintentional ``hyperparameter hacking'' in the comparisons between \name\ aggregators, we sweep over the same set of hyperparameters for all \name\ variants (choosing the best setting for each variant according to performance on a validation set). 
The set of possible hyperparameter values was determined on early validation tests using subsets of the citation and Reddit data that we then discarded from our analyses. 
The appendix contains further implementation details.\footnote{Code and links to the datasets:  \url{http://snap.stanford.edu/graphsage/}}

\subsection{Inductive learning on evolving graphs: Citation and Reddit data}
\label{sec:evolving}

\begin{table}[t!]
\centering
\caption{Prediction results for the three datasets (micro-averaged F1 scores). Results for unsupervised and fully supervised \name\ are shown. Analogous trends hold for macro-averaged scores.}
\label{tab:main_results}
{\small
\begin{tabular}{lcccccc}
\toprule
&\multicolumn{2}{c}{Citation}& \multicolumn{2}{c}{Reddit} & \multicolumn{2}{c}{PPI} \\
\cmidrule(r{1em}l{1em}){2-3}  \cmidrule(r{1em}l{1em}){4-5} \cmidrule(r{1em}l{1em}){6-7}
Name & Unsup. F1 &  Sup. F1 & Unsup. F1 & Sup. F1 &  Unsup. F1 & Sup. F1\\
\midrule
Random & 0.206 & 0.206 & 0.043 & 0.042 & 0.396 & 0.396 \\
Raw features & 0.575 & 0.575 & 0.585 & 0.585 & 0.422 & 0.422 \\
DeepWalk & 0.565 &  0.565 & 0.324 & 0.324 &  --- & --- \\\vspace{3pt}
DeepWalk + features & 0.701 &   0.701 & 0.691 & 0.691 & --- & --- \\
\name-GCN & 0.742 & 0.772 & {\bf 0.908} &  0.930 & 0.465 &  0.500\\
\name-mean & 0.778 & 0.820 & 0.897 &  0.950 & 0.486 & 0.598 \\
\name-LSTM & 0.788 & 0.832 & {\bf 0.907} & {\bf 0.954} & 0.482 &  {\bf 0.612}\\
\name-pool & {\bf 0.798} & {\bf 0.839} & 0.892 & 0.948 &  {\bf 0.502} & 0.600\\
\midrule
\% gain over feat. & 39\% & 46\% & 55\% & 63\% & 19\% & 45\%\\
\bottomrule
\end{tabular}
}
\end{table}
\begin{figure}
\vspace{-10pt}
\begin{subfigure}[t]{0.48\textwidth}
\centering
\includegraphics[width=0.99\textwidth]{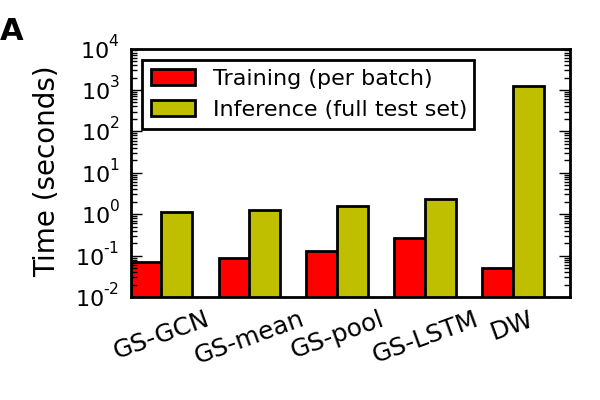}
\end{subfigure}
\begin{subfigure}[t]{0.52\textwidth}
\centering
\includegraphics[width=0.99\textwidth]{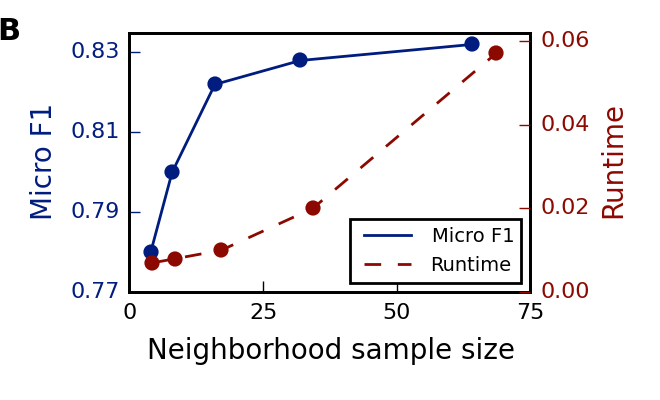}
\end{subfigure}
\caption{\textbf{A}: Timing experiments on Reddit data, with training batches of size 512 and inference on the full test set (79,534 nodes). \textbf{B}: Model performance with respect to the size of the sampled neighborhood, where the ``neighborhood sample size'' refers to the number of neighbors sampled at each depth for $K=2$ with  $S_1=S_2$ (on the citation data using \name-mean). }
\label{fig:time_and_noise}
\end{figure}

Our first two experiments are on classifying nodes in evolving information graphs, a task that is especially relevant to high-throughput production systems, which constantly encounter unseen data. 

\xhdr{Citation data}
Our first task is predicting paper subject categories on a large citation dataset. 
We use an undirected citation graph dataset derived from the Thomson Reuters Web of Science Core Collection, corresponding to all papers in six biology-related fields for the years 2000-2005.
The node labels for this dataset correspond to the six different field labels. 
In total, this is dataset contains 302,424 nodes with an average degree of 9.15.
We train all the algorithms on the 2000-2004 data and use the 2005 data for testing (with 30\% used for validation). 
For features, we used node degrees and processed the paper abstracts according Arora et al.'s~\cite{arora2017simple} sentence embedding approach, with 300-dimensional word vectors trained using the GenSim word2vec implementation \cite{rehurek_lrec}.

\xhdr{Reddit data}
In our second task, we predict which community different Reddit posts belong to. 
Reddit is a large online discussion forum where users post and comment on content in different topical communities. 
We constructed a graph dataset from Reddit posts made in the month of September, 2014. 
The node label in this case is the community, or ``subreddit'', that a post belongs to.  
We sampled 50 large communities and built a post-to-post graph, connecting posts if the same user comments on both. 
In total this dataset contains 232,965 posts with an average degree of 492.
We use the first 20 days for training and the remaining days for testing (with 30\% used for validation).  
For features, we use off-the-shelf 300-dimensional GloVe CommonCrawl word vectors \cite{pennington2014glove}; for each post, we concatenated (i) the average embedding of the post title, (ii) the average embedding of all the post's comments (iii) the post's score, and (iv) the number of comments made on the post. 

The first four columns of Table \ref{tab:main_results} summarize the performance of \name\, as well as the baseline approaches on these two datasets.
We find that \name\ outperforms all the baselines by a significant margin, and the trainable, neural network aggregators provide significant gains compared to the GCN approach. 
For example, the unsupervised variant \name-pool outperforms the concatenation of the DeepWalk embeddings and the raw features by 13.8\% on the citation data and 29.1\% on the Reddit data, while the supervised version provides a gain of 19.7\% and 37.2\%, respectively. 
Interestingly, the LSTM based aggregator shows strong performance, despite the fact that it is designed for sequential data and not unordered sets. 
Lastly, we see that the performance of unsupervised \name\ is reasonably competitive with the fully supervised version, indicating that our framework can achieve strong performance without task-specific fine-tuning. 
 
\subsection{Generalizing across graphs: Protein-protein interactions}\label{sec:ppi}

We now consider the task of generalizing across graphs, which requires learning about node roles rather than community structure. 
We classify protein roles---in terms of their cellular functions from gene ontology---in various
protein-protein interaction (PPI) graphs, with each graph corresponding to a different human tissue \cite{zitnik2017tissue}.
We use positional gene sets, motif gene sets and immunological signatures as features and gene
ontology sets as labels (121 in total), collected from the Molecular Signatures Database \cite{subramanian2005gene}.
The average graph contains 2373 nodes, with an average degree of 28.8.
We train all algorithms on 20 graphs and then average prediction F1 scores on two test graphs (with two other graphs used for validation). 

The final two columns of Table \ref{tab:main_results} summarize the accuracies of the various approaches on this data.
Again we see that \name\ significantly outperforms the baseline approaches, with the LSTM- and pooling-based aggregators providing substantial gains over the mean- and GCN-based aggregators.\footnote{Note that in very recent follow-up work Chen and Zhu \cite{chen2017stochastic} achieve superior performance by optimizing the \name\ hyperparameters specifically for the PPI task and implementing new training techniques (e.g., dropout, layer normalization, and a new sampling scheme). We refer the reader to their work for the current state-of-the-art numbers on the PPI dataset that are possible using a variant of the \name\ approach.}

\subsection{Runtime and parameter sensitivity}\label{sec:sensitivity}

Figure \ref{fig:time_and_noise}.A summarizes the training and test runtimes for the different approaches. 
The training time for the methods are comparable (with \name-LSTM being the slowest). 
However, the need to sample new random walks and run new rounds of SGD to embed unseen nodes makes DeepWalk $100\text{-}500\times$ slower at test time. 

For the \name\ variants, we found that setting $K=2$ provided a consistent boost in accuracy of around $10\text{-}15\%$, on average, compared
to $K=1$; however, increasing $K$ beyond 2 gave marginal returns in performance ($0\text{-}5\%$)
while increasing the runtime by a prohibitively large factor of $10\text{-}100{\times}$, depending
on the neighborhood sample size.
We also found diminishing returns for sampling large neighborhoods (Figure
\ref{fig:time_and_noise}.B).
Thus, despite
the higher variance induced by sub-sampling neighborhoods, \name\ is still able to maintain strong predictive accuracy, while significantly improving the runtime. 

\subsection{Summary comparison between the different aggregator architectures}\label{sec:sensitivity}

Overall, we found that the LSTM- and pool-based aggregators performed the best, in terms of both average performance and number of experimental settings where they were the top-performing method (Table \ref{tab:main_results}).
To give more quantitative insight into these trends, we consider each of the six different experimental settings (i.e., $\textrm{(3 datasets)} \times (\textrm{unsupervised vs.\ supervised})$) as trials and consider what performance trends are likely to generalize.
In particular, we use the non-parametric Wilcoxon Signed-Rank Test \cite{siegal1956nonparametric} to quantify the differences between the different aggregators across trials, reporting the $T$-statistic and $p$-value where applicable. 
Note that this method is rank-based and essentially tests whether we would expect one particular approach to outperform another in a new experimental setting.
Given our small sample size of only 6 different settings, this significance test is somewhat underpowered; nonetheless, the $T$-statistic and associated $p$-values are useful quantitative measures to assess the aggregators' relative performances. 

We see that LSTM-, pool- and mean-based aggregators all provide statistically significant gains over the GCN-based approach ($T=1.0$, $p=0.02$ for all three). 
However, the gains of the LSTM and pool approaches over the mean-based aggregator are more marginal ($T=1.5$, $p=0.03$, comparing LSTM to mean; $T=4.5$, $p=0.10$, comparing pool to mean).
There is no significant difference between the LSTM and pool approaches ($T=10.0$, $p=0.46$). 
However, \name-LSTM is significantly slower than \name-pool (by a factor of ${\approx}2{\times}$), perhaps giving the pooling-based aggregator a slight edge overall.

\section{Theoretical analysis}\label{sec:theory}

In this section, we probe the expressive capabilities of \name\ in order to 
provide insight into how \name\ can learn about graph structure, even though it is inherently based on features. 
As a case-study, we consider whether \name\ can learn to predict the clustering coefficient of a node, i.e., the proportion of triangles that are closed within the node's 1-hop neighborhood \cite{watts1998collective}.  The clustering coefficient is a popular  measure of how clustered a node's local neighborhood is, and it serves as a building block for many more complicated structural motifs \cite{benson2016higher}. We can show that Algorithm \ref{alg:basic} is capable of approximating clustering coefficients to an arbitrary degree of precision:

\begin{theorem}\label{thm:main}
Let $\mb{x}_v \in U, \forall v \in \mathcal{V}$ denote the feature inputs for Algorithm \ref{alg:basic} on graph $\mathcal{G}=(\mathcal{V}, \mathcal{E})$, where $U$ is any compact subset of $\mathbb{R}^d$. Suppose that there exists a fixed positive constant $C \in \mathbb{R}^+$ such that $\|\mb{x}_v  - \mb{x}_{v'}\|_2 > C$ for all pairs of nodes. Then we have that $\forall \epsilon > 0$ there exists a parameter setting $\mb{\Theta}^*$ for Algorithm \ref{alg:basic} such that after $K=4$ iterations
$$|{z}_v - c_v| < \epsilon , \forall v \in \mathcal{\V},$$ where $z_v \in \mathbb{R}$ are final output values generated by Algorithm 1 and $c_v$ are node clustering coefficients.\cut{, as defined in \cite{watts1998collective}.}
\end{theorem}

Theorem 1 states that for any graph there exists a parameter setting for Algorithm 1 such that it can approximate clustering coefficients in that graph to an arbitrary precision, if the features for every node are distinct (and if the model is sufficiently high-dimensional). 
The full proof of Theorem 1 is in the Appendix. 
Note that as a corollary of Theorem 1, \name\ can learn about local graph structure, even when the node feature inputs are sampled from an absolutely continuous random distribution (see the Appendix for details).
The basic idea behind the proof is that if each node has a unique feature representation, then we can learn to map nodes to indicator vectors and identify node neighborhoods.
The proof of Theorem 1 relies on some properties of the pooling aggregator, which also provides insight into why \name-pool outperforms the GCN and mean-based aggregators. 

%% file: 060conclusion.tex

\section{Conclusion}

We introduced a novel approach that allows embeddings to be efficiently generated for unseen nodes. 
\name\ consistently outperforms state-of-the-art baselines, effectively trades off performance and runtime by sampling node neighborhoods, and our theoretical analysis provides insight into how our approach can learn about local graph structures. 
A number of extensions and potential improvements are possible, such as extending \name\ to incorporate directed or multi-modal graphs. 
A particularly interesting direction for future work is exploring non-uniform neighborhood sampling functions, and perhaps even learning these functions as part of the \name\ optimization.

%% file: 070appendix.tex

\section{Minibatch pseudocode}\label{sec:minibatch}

In order to use stochastic gradient descent, we adapt our algorithm to allow forward and backward
propagation for minibatches of nodes and edges. 
Here we focus on the minibatch forward propagation
algorithm, analogous to Algorithm \ref{alg:basic}.
In the forward propagation of \name\, the minibatch $\mathcal{B}$ contains nodes that we want to generate representations for. 
Algorithm 2 gives the pseudocode for the minibatch approach. 

\setcounter{algocf}{1}
\begin{algorithm}
\caption{\name\ minibatch forward propagation algorithm}
	\SetKwInOut{Input}{Input}\SetKwInOut{Output}{Output}
    \Input{~Graph $\G(\V,\E)$;\\
     input features $\{\mb{x}_v, \forall v\in \mathcal{B}\}$; \\
     depth $K$; weight matrices $\mb{W}^{k}, \forall k \in \{1,...,K\}$; \\
     non-linearity $\sigma$; \\
     differentiable aggregator functions $\textsc{aggregate}_k, \forall k \in \{1,...,K\}$;\\
      neighborhood sampling functions, $\mathcal{N}_k : v \rightarrow 2^{\V}, \forall k \in \{1,...,K\}$}
    \Output{~Vector representations $\mb{z}_v$ for all $v \in \mathcal{B}$}
    \BlankLine
         $\mathcal{B}^K \leftarrow \mathcal{B}$\;
        \For{$k=K...1$}{
            $B^{k-1} \leftarrow \mathcal{B}^{k}$ \;
            \For{$u \in \mathcal{B}^{k}$}{
                $\mathcal{B}^{k-1} \leftarrow \mathcal{B}^{k-1} \cup \mathcal{N}_k(u)$\;
            }
    	}
        $\mb{h}^0_u \leftarrow \mb{x}_v, \forall v \in \mathcal{B}^0$ \;
        \For{$k=1...K$}{
            \For{$u \in \mathcal{B}^k$}{
    	        $\mb{h}^{k}_{\mathcal{N}(u)} \leftarrow \textsc{aggregate}_k(\{\mb{h}_{u'}^{k-1},
                    \forall u' \in \mathcal{N}_k(u)\})$\;
    	        $\mb{h}^k_u \leftarrow \sigma\left(\mb{W}^{k}\cdot\textsc{concat}(\mb{h}_u^{k-1}, \mb{h}^{k}_{\mathcal{N}(u)})\right)$\;
    	         $\mb{h}^{k}_u\leftarrow \mb{h}^{k}_u/ \|\mb{h}^{k}_u\|_2$\;
            }
        }
     $\mb{z}_u\leftarrow \mb{h}^{K}_u, \forall u \in \mathcal{B}$ 
\end{algorithm}

The main idea is to sample all the nodes needed for the computation first. 
Lines 2-7 of Algorithm 2 correspond to the sampling stage. 
Each set $\mathcal{B}^k$ contains the nodes that are needed to compute the representations of
nodes $v \in \mathcal{B}^{k+1}$, i.e., the nodes in the $(k+1)$-st iteration, or ``layer'', of Algorithm 1.
Lines 9-15 correspond to the aggregation stage, which is almost identical to the batch inference
algorithm. 
Note that in Lines 12 and 13, the representation at iteration $k$ of any node in set $\mathcal{B}^k$ can be computed, because its representation at iteration $k-1$  and the representations of its sampled neighbors at iteration $k-1$ have already been computed
in the previous loop.
The algorithm thus avoids computing the representations for nodes that are not in the current
minibatch and not used during the current iteration of stochastic gradient descent.
We use the notation $\mathcal{N}_k(u)$ to denote a deterministic function which specifies a random sample of a node's neighborhood (i.e., the randomness is assumed to be pre-computed in the mappings).
We index this function by $k$ to denote the fact that the random samples are independent across iterations over $k$.
We use a uniform sampling function in this work and sample with replacement in cases where the sample size is larger than the node's degree. 

Note that the sampling process in Algorithm 2 is conceptually reversed compared to the iterations over $k$ in Algorithm \ref{alg:basic}: we start with the ``layer-K'' nodes (i.e., the nodes in $\mathcal{B}$) that we want to generate representations for; then we sample their neighbors (i.e., the nodes at ``layer-K-1'' of the algorithm) and so on.
One consequence of this is that the definition of neighborhood sampling sizes can be somewhat counterintuitive.
In particular, if we use $K=2$ total iterations with sample sizes $S_1$ and $S_2$ then this means that we sample $S_1$ nodes during iteration $k=1$ of Algorithm 1 and $S_2$ nodes during iteration $k=2$, and---from the perspective of the ``target'' nodes in $\mathcal{B}$ that we want to generate representations for after iteration $k=2$---this amounts to sampling $S_2$ of their immediate neighbors and $S_1\cdot S_2$ of their 2-hop neighbors.

\section{Additional Dataset Details}

In this section, we provide some additional, relevant dataset details.
The full PPI and Reddit datasets are available at: \url{http://snap.stanford.edu/graphsage/}.
The Web of Science dataset (WoS) is licensed by Thomson Reuters and can be made available to groups with valid WoS licenses. 

\paragraph{Reddit data}

To sample communities, we ranked communities by their total number of comments in 2014 and selected the communities with ranks [11,50] (inclusive).
We omitted the largest communities because they are large, generic default communities that substantially skew the class distribution.  
We selected the largest connected component of the graph defined over the union of these communities. 
We performed early validation experiments and model development on data from October and November, 2014. 

Details on the source of the Reddit data are at: \url{https://archive.org/details/FullRedditSubmissionCorpus2006ThruAugust2015} and
\url{https://archive.org/details/2015_reddit_comments_corpus}.

\paragraph{WoS data}

We selected the following subfields manually, based on them being of relatively equal size and all biology-related fields.
We performed early validation and model development on the neuroscience subfield (code=RU, which is excluded from our final set). 
We did not run any experiments on any other subsets of the WoS data. 
We took the largest connected component of the graph defined over the union of these fields. 
\begin{itemize}
\item
Immunology (code: NI, number of documents: 77356)
\item
Ecology (code: GU, number of documents: 37935)
\item
Biophysics (code: DA, number of documents: 36688)
\item
Endocrinology and Metabolism (code: IA, number of documents: 52225). 
\item
Cell Biology (code: DR, number of documents: 84231)
\item
Biology (other) (code: CU, number of documents: 13988)
\end{itemize}

\paragraph{PPI Tissue Data}

For training, we randomly selected 20 PPI networks that had at least 15,000 edges. For testing and validation, we selected 4 large networks (2 for validation, 2 for testing, each with at least 35,000 edges). All experiments for model design and development were performed on the same 2 validation networks, and we used the same random training set in all experiments.

We selected features that included at least 10\% of the proteins that appear in any of the PPI graphs. 
Note that the feature data is very sparse for dataset ($42\%$ of nodes have no non-zero feature values), which makes leveraging neighborhood information critical. 

\section{Details on the Experimental Setup and Hyperparameter Tuning}\label{sec:hyper}

\paragraph{Random walks for the unsupervised objective}
For all settings, we ran 50 random walks of length 5 from each node in order to obtain the pairs needed for the unsupervised loss (Equation \ref{eq:neg}).
Our implementation of the random walks is in pure Python and is based directly on Python code provided by Perozzi et al.~\cite{perozzi2014deepwalk}.

\paragraph{Logistic regression model}
For the feature only model and to make predictions on the embeddings output from the unsupervised models, 
we used the logistic SGDClassifier from the scikit-learn Python package \cite{scikit-learn}, with all default settings.
Note that this model is always optimized only on the training nodes and it is not fine-tuned on the embeddings that are generated for the test data.

\paragraph{Hyperparameter selection}

In all settings, we performed hyperparameter selection on the learning rate and the model dimension.
With the exception of DeepWalk, we performed a parameter sweep on initial learning rates $\{0.01, 0.001, 0.0001\}$ for the supervised models and $\{2\times10^{-6}, 2\times10^{-7}, 2\times10^{-8}\}$ for the unsupervised models.\footnote{Note that these values differ from our previous reported pre-print values because they are corrected to account for an extraneous normalization by the batch size. We thank Ben Johnson for pointing out this discrepancy.}
When applicable, we tested a ``big'' and ``small'' version of each model, where we tried to keep the overall model sizes comparable.
For the pooling aggregator, the ``big'' model had a pooling dimension of 1024, while the ``small'' model had a dimension of  512.
For the LSTM aggregator, the ``big'' model had a hidden dimension of 256, while the ``small'' model had a hidden dimension of 128; note that the actual parameter count for the LSTM is roughly $4{\times}$ this number, due to weights for the different gates. 
In all experiments and for all models we specify the output dimension of the $\mb{h}^k_i$ vectors at every depth $k$ of the recursion to be $256$. 
All models use rectified linear units as a non-linear activation function.
All the unsupervised \name\ models and DeepWalk used 20 negative samples with context distribution smoothing over node degrees using a smoothing parameter of $0.75$, following \cite{grover2016node2vec,mikolov2013distributed,perozzi2014deepwalk}.
Initial experiments revealed that DeepWalk performed much better with large learning rates, so we swept over rates in the set $\{0.2,0.4,0.8\}$. 
For the supervised \name\ methods, we ran 10 epochs for all models. 
All methods except DeepWalk use batch sizes of 512.
We found that DeepWalk achieved faster wall-clock convergence with a smaller batch size of 64. 

\paragraph{Hardware}
Except for DeepWalk, we ran experiments single a machine with 4 NVIDIA Titan X Pascal GPUs (12Gb of RAM at 10Gbps speed), 16 Intel Xeon CPUs (E5-2623 v4 @ 2.60GHz), and 256Gb of RAM.  
DeepWalk was faster on a CPU intensive machine with 144 Intel Xeon CPUs (E7-8890 v3 @ 2.50GHz) and 2Tb of RAM. 
Overall, our experiments took about 3 days in a shared resource setting. 
We expect that a consumer-grade single-GPU machine (e.g., with a Titan X GPU) could complete our full set of experiments in 4-7 days, if its full resources were dedicated. 

\paragraph{Notes on the DeepWalk implementation}
Existing DeepWalk implementations \cite{perozzi2014deepwalk,grover2016node2vec} are simply wrappers around dedicated word2vec code, and they do not easily support embedding new nodes and other variations. 
Moreover, this makes it difficult to compare runtimes and other statistics for these approaches.
For this reason, we reimplemented DeepWalk in pure TensorFlow, using the vector initializations etc that are described in the TensorFlow word2vec tutorial.\footnote{\url{https://github.com/tensorflow/models/blob/master/tutorials/embedding/word2vec.py}}

We found that DeepWalk was much slower to converge than the other methods, and since it is 2-5X faster at training, we gave it 5 passes over the random walk data, instead of one.
To update the DeepWalk method on new data, we ran 50 random walks of length 5 (as described above) and performed updates on the embeddings for the new nodes while holding the already trained embeddings fixed. 
We also tested two variants, one where we restricted the sampled random walk ``context nodes'' to only be from the set of already trained nodes (which alleviates statistical drift) and an approach without this restriction. 
We always selected the better performing variant. 
Note that despite DeepWalk's poor performance on the inductive task, it is far more competitive when tested in the transductive setting, where it can be extensively trained on a single, fixed graph. 
(That said, Kipf et al \cite{kipf2016semi}\cite{kipf2016variational} found that GCN-based approach consistently outperformed DeepWalk, even in the transductive setting on link prediction, a task that theoretically favors DeepWalk.)
We did observe DeepWalk's performance {\em could} improve with further training, and in some cases it could become competitive with the unsupervised \name\ approaches (but not the supervised approaches) if we let it run for ${>}1000{\times}$ longer than the other approaches (in terms of wall clock time for prediction on the test set); however, we did not deem this to be a meaningful comparison for the inductive task. 

Note that DeepWalk is also equivalent to the node2vec model \cite{grover2016node2vec} with $p=q=1$. 

\paragraph{Notes on neighborhood sampling}
Due to the heavy-tailed nature of degree distributions we downsample the edges in all graphs before feeding them into the \name\ algorithm.
In particular, we subsample edges so that no node has degree larger than $128$. 
Since we only sample at most $25$ neighbors per node, this is a reasonable tradeoff.
This downsampling allows us to store neighborhood information as dense adjacency lists, which drastically improves computational efficiency. 
For the Reddit data we also downsampled the edges of the original graph as a pre-processing step, since the original graph is extremely dense. 
All experiments are on the downsampled version, but we release the full version on the project website for reference.

\section{Alignment Issues and Orthogonal Invariance for DeepWalk and Related Approaches}\label{sec:rotational}

DeepWalk \cite{perozzi2014deepwalk}, node2vec \cite{grover2016node2vec}, and other recent successful node embedding approaches employ objective functions of the form:
\begin{equation}\label{eq:n2vobj}
\alpha\sum_{i,j \in \mathcal{A}}f(\mb{z}_i^\top\mb{z}_{j}) + \beta\sum_{i,j \in \mathcal{B}}g(\mb{z}_i^\top\mb{z}_{j}) 
\end{equation}
where $f$, $g$ are smooth, continuous functions, $\mb{z}_i$ are the node representations that are being directly optimized (i.e., via embedding look-ups), and $\mathcal{A}, \mathcal{B}$ are sets of pairs of nodes. 
Note that in many cases, in the actual code implementations used by the authors of these approaches, nodes are associated with two unique embedding vectors and the arguments to the dot products in $f$ and $g$ are drawn for distinct embedding look-ups (e.g., \cite{grover2016node2vec,perozzi2014deepwalk}); however, this does not fundamentally alter the learning algorithm. 
The majority of approaches also normalize the learned embeddings to unit length, so we assume this post-processing as well.

By connection to word embedding approaches and the arguments of \cite{levy2014neural}, these approaches can also be viewed as stochastic, implicit matrix factorizations where we are trying to learn a matrix $\mb{Z} \in \mathbb{R}^{|\V| \times d}$ such that
\begin{equation}
\mb{Z}\mb{Z}^\top \approx \mb{M},
\end{equation}
where $\mb{M}$ is some matrix containing random walk statistics. 

An important consequence of this structure is that the embeddings can be rotated by an arbitrary orthogonal matrix, without impacting the objective:
\begin{equation}
\mb{Z}\mb{Q}^\top\mb{Q}\mb{Z}^\top = \mb{Z}\mb{Z}^\top,
\end{equation}
where $\mb{Q} \in \mathbb{R}^{d\times d}$ is any orthogonal matrix. 
Since the embeddings are otherwise unconstrained and the only error signal comes from the orthogonally-invariant objective \eqref{eq:n2vobj}, the entire embedding space is free to arbitrarily rotate during training.  

Two clear consequences of this are:
\begin{enumerate}
\item
	Suppose we run an embedding approach based on \eqref{eq:n2vobj} on two separate graphs A and B using the same output dimension. Without some explicit penalty enforcing alignment, the learned embeddings spaces for the two graphs will be arbitrarily rotated with respect to each other after training. 
	Thus, for any node classification method that is trained on individual embeddings from graph A, inputting the embeddings from graph B will be essentially random. 
	This fact is also simply true by virtue of the fact that the $\mb{M}$ matrices of these graphs are completely disjoint. 
	Of course, if we had a way to match ``similar'' nodes between the graphs, then it could be possible to use an alignment procedure to share information between the graphs, such as the procedure proposed by \cite{hamilton2016diachronic} for aligning the output of word embedding algorithms. 
	Investigating such alignment procedures is an interesting direction for future work; though these approaches will inevitably be slow run on new data, compared to approaches like \name\ that can simply generate embeddings for new nodes without any additional training or alignment. 
\item
  Suppose that we run an embedding approach based on \eqref{eq:n2vobj} on graph C at time $t$ and train a classifier on the learned embeddings. Then at time $t+1$ we add more nodes to C and run a new round of SGD and update all embeddings.
  Two issues arise: First by analogy to point 1 above, if the new nodes are only connected to a very small number of the old nodes, then the embedding space for the new nodes can essentially become rotated with respect to the original embedding space.
  Moreover, if we update all embeddings during training (not just for the new nodes), as suggested by \cite{perozzi2014deepwalk}'s streaming approach to DeepWalk, then the embedding space can arbitrarily rotate compared to the embedding space that we trained our classifier on, which only further exasperates the problem.  
\end{enumerate}

Note that this rotational invariance is not problematic for tasks that only rely on pairwise node distances (e.g., link prediction via dot products). 
Moreover, some reasonable approaches to alleviate this issue of statistical drift are to (1) not update the already trained embeddings when optimizing the embeddings for new test nodes and (2) to only keep existing nodes as ``context nodes'' in the sampled random walks, i.e. to ensure that every dot-product in the skip-gram objective is the product of an already-trained node and a new/test node.
We tried both of these approaches in this work and always selected the best performing DeepWalk variant. 

Also note that empirically DeepWalk performs better on the citation data than the Reddit data (Section \ref{sec:evolving}) because this statistical drift is worse in the Reddit data, compared to the citation graph.
In particular, the Reddit data has fewer edges from the test set to the train set, which help prevent mis-alignment: 96\% of the 2005 citation links connect back to the 2000-2004 data, while only 73\% of edges in the Reddit test set connect back to the train data. 

\section{Proof of Theorem 1}

\newcommand{\R}{\mathbb{R}}
To prove Theorem 1, we first prove three lemmas:
\begin{itemize}
\item
Lemma 1 states that there exists a continuous function that is guaranteed to only be positive in closed balls around a fixed number of points, with some noise tolerance. 
\item
Lemma 2 notes that we can approximate the function in Lemma 1 to an arbitrary precision using a multilayer perceptron with a single hidden layer.
\item
Lemma 3 builds off the preceding two lemmas to prove that the pooling architecture can learn to map nodes to unique indicator vectors, assuming that all the input feature vectors are sufficiently distinct. 
\end{itemize}
We also rely on fact that the max-pooling operator (with at least one hidden layer) is capable of approximating any Hausdorff continuous, symmetric function to an arbitrary $\epsilon$ precision \cite{qi2016pointnet}.

We note that all of the following are essentially {\em identifiability} arguments. 
We show that there exists a parameter setting for which Algorithm \ref{alg:basic} can learn nodes clustering coefficients, which is non-obvious given that it operates by aggregating feature information.
The {\em efficient learnability} of the functions described is the subject of future work. 
We also note that these proofs are conservative in the sense that clustering coefficients may be in fact identifiable in fewer iterations, or with less restrictions, than we impose. 
Moreover, due to our reliance on two universal approximation theorems \cite{hornik1991approximation,qi2016pointnet}, the required dimensionality is in principle $O(|\V|)$.
We can provide a more informative bound on the required output dimension of some particular layers (e..g., Lemma 3); however, in the worst case this identifiability argument relies on having a dimension of $O(|\V|)$. 
It is worth noting, however,  that Kipf et al's ``featureless'' GCN approach has parameter dimension $O(|\V|)$, so this requirement is not entirely unreasonable \cite{kipf2016semi,kipf2016variational}.

Following Theorem 1, we let $\mb{x}_v \in U, \forall v \in \V$ denote the feature inputs for Algorithm \ref{alg:basic} on graph $\mathcal{G}=(\mathcal{V}, \mathcal{E})$, where $U$ is any compact subset of $\mathbb{R}^d$.

\begin{lemma}
Let $C \in \mathbb{R}^+$ be a fixed positive constant. 
Then for any non-empty finite subset of nodes $\mathcal{D} \subseteq \mathcal{V}$, there exists a continuous function $g : U \rightarrow \R$ such that 
\begin{equation}
\begin{cases}
g(\mb{x}) > \epsilon,& \textrm{if $\|\mb{x}-\mb{x}_v\|_2 = 0$ for some $v \in \mathcal{D}$}\\
g(\mb{x}) \leq {-}\epsilon,&\textrm{if $\|\mb{x}-\mb{x}_v\|_2 > C, \forall v \in \mathcal{D}$},
\end{cases}
\end{equation}
where $\epsilon < 0.5$ is a chosen error tolerance. 
\end{lemma}
\begin{proof}
Many such functions exist. For concreteness, we provide one construction that satisfies these criteria. 
Let $\mb{x} \in U$ denote an arbitrary input to $g$, let $d_v = \|\mb{x} - \mb{x}_v\|_2, \forall v \in \mathcal{D}$, and let $g$ be defined as $g(\mb{x}) = \sum_{v \in \mathcal{D}} g_v(\mb{x})$ with
\newcommand{\setsize}{|\mathcal{D}|}
\begin{equation}
g_v(\mb{x}) =  \frac{3\setsize\epsilon}{bd_v^2 + 1} - 2\epsilon
\end{equation}
where $b = \frac{3\setsize - 1}{C^2}>0$.
By construction:
\begin{enumerate}
\item
	$g_v$ has a unique maximum of $3\setsize\epsilon-2\epsilon>2\setsize\epsilon$ at $d_v = 0$. 
\item
	$\lim_{d_v \rightarrow \infty}\left(  \frac{3\setsize\epsilon}{bd_v^2 + 1} - 2\epsilon \right) = -2\epsilon$
\item 
   $\frac{3\setsize\epsilon}{bd_v^2 + 1} - 2\epsilon \leq -\epsilon$ if $d_v \geq C$. 
\end{enumerate}
Note also that $g$ is continuous on its domain ($d_v \in \mathbb{R}^+$) since it is the sum of finite set of continuous functions. 
Moreover, we have that, for a given input $\mb{x} \in U$, if $d_v\geq C$ for all points $v \in \mathcal{D}$ then $g(\mb{x}) = \sum_{v \in \mathcal{D}} g_v(\mb{a}) \leq {-}\epsilon$ by property 3 above. 
And, if $d_v=0$ for any $v \in \mathcal{D}$, then $g$ is positive by construction, by properties 1 and 2, since in this case,
\begin{align*}
g_v(\mb{x}) + \sum_{v' \in \mathcal{D}\setminus v}g_{v'}(\mb{x}) &\geq g_v(\mb{x}) - (\setsize -1)2\epsilon\\
&> g_v(\mb{x}) - 2(\setsize)\epsilon\\
&>2(\setsize)\epsilon - 2(\setsize)\epsilon\\
& >0,
\end{align*}
 so we know that $g$ is positive whenever $d_v = 0$ for any node and negative whenever $d_v > C$ for all nodes.  
\end{proof}

\begin{lemma}
The function $g : U \rightarrow \mathbb{R}$ can be approximated to an arbitrary degree of precision by standard multilayer perceptron (MLP) with least one hidden layer and a non-constant monotonically increasing activation function (e.g., a rectified linear unit).  In precise terms, if we let $f_{\theta_\sigma}$ denote this MLP and $\theta_\sigma$ its parameters, we have that $\forall \epsilon$, $\exists \theta_\sigma$ such that 
$|f_{\theta_\sigma}(\mb{x}) - g(\mb{x})| < \epsilon|, \forall \mb{x} \in U$.
\end{lemma}
\begin{proof}
This is a direct consequence of Theorem 2 in \cite{hornik1991approximation}.
\end{proof}

\begin{lemma}
Let $\mb{A}$ be the adjacency matrix of $G$, let $\mathcal{N}^2(v)$ denote the 2-hop neighborhood of a node, $v$, and define $\chi(G^4)$ as the chromatic number of the graph with adjacency matrix $\mb{A}^4$ (ignoring self-loops). Suppose that there exists a fixed positive constant $C \in \mathbb{R}^+$ such that $\|\mb{x}_v  - \mb{x}_{v'}\|_2 > C$ for all pairs of nodes. Then we have that there exists a parameter setting for Algorithm \ref{alg:basic}, using a pooling aggregator at depth $k=1$, where this pooling aggregator has $\geq 2$ hidden layers with rectified non-linear units, such that 
$$\mb{h}^1_v \neq \mb{h}^1_{v'}, \forall (v,v') \in \{(v,v') : \exists u \in \V, v,v' \in \mathcal{N}^2(u)\}, \mb{h}^1_v, \mb{h}^1_{v'} \in \mathcal{E}^{\chi(\G^4)}_I$$
where $\mathcal{E}^{\chi(\G^4)}_I$ is the set of one-hot indicator vectors of dimension $\mathcal{\chi}(\G^4)$.
\end{lemma}
\begin{proof}
By the definition of the chromatic number, we know that we can label every node in $\mathcal{V}$ using $\chi(\G^4)$ unique colors, such that no two nodes that co-occur in any node's 2-hop neighborhood are assigned the same color. 
Thus, with exactly $\chi(\G^4)$ dimensions we can assign a unique one-hot indicator vector to every node, where no two nodes that co-occur in any 2-hop neighborhood have the same vector. 
In other words, each color defines a subset of nodes $\mathcal{D} \subseteq \mathcal{V}$ and this subset of nodes can all be mapped to the same indicator vector without introducing conflicts. 

By Lemma 1 and 2 and the assumption that $\|\mb{x}_v  - \mb{x}_{v'}\|_2 > C$ for all pairs of nodes, we can choose an $\epsilon < 0.5$ and there exists a single-layer MLP, $f_{\theta_\sigma}$, such that for any subset of nodes $\mathcal{D} \subseteq \mathcal{V}$:
\begin{equation}
\begin{cases}
f_{\theta_\sigma}(\mb{x}_v) > 0 ,&\forall v \in \mathcal{D}\\
f_{\theta_\sigma}(\mb{x}_v) < 0,& \forall v \in \V \setminus \mathcal{D}.
\end{cases}
\end{equation}
By making this MLP one layer deeper and specifically using a rectified linear activation function, we can return a positive value only for nodes in the subset $\mathcal{D}$ and zero otherwise, and, since we normalize after applying the aggregator layer, this single positive value can be mapped to an indicator vector. 
Moreover, we can create $\chi(\G^4)$ such MLPs, where each MLP corresponds to a different color/subset; equivalently each MLP corresponds to a different max-pooling dimension in equation \ref{eq:pool} of the main text. 
\end{proof}

We now restate Theorem 1 and provide a proof. 

\renewcommand{\thetheorem}{\ref{thm:main}}
\begin{theorem}
Let $\mb{x}_v \in \mathbb{R}^d, \forall v \in \mathcal{V}$ denote the feature inputs for Algorithm \ref{alg:basic} on graph $\mathcal{G}=(\mathcal{V}, \mathcal{E})$, where $U$ is any compact subset of $\mathbb{R}^d$. Suppose that there exists a fixed positive constant $C \in \mathbb{R}^+$ such that $\|\mb{x}_v  - \mb{x}_{v'}\|_2 > C$ for all pairs of nodes. Then we have that $\forall \epsilon > 0$ there exists a parameter setting $\mb{\Theta}^*$ for Algorithm \ref{alg:basic} such that after $K=4$ iterations 
$$|{z}_v - c_v| < \epsilon , \forall v \in \mathcal{\V},$$ where $z_v \in \mathbb{R}$ are final output values generated by Algorithm 1 and $c_v$ are node clustering coefficients, as defined in \cite{watts1998collective}. 
\end{theorem}

\begin{proof}
Without loss of generality, we describe how to compute the clustering coefficient for an arbitrary node $v$.
For notational convenience we use $\oplus$ to denote vector concatenation and $d_v$ to denote the degree of node $v$.
This proof requires 4 iterations of Algorithm 1, where we use the pooling aggregator at all depths. 
For clarity and we ignore issues related to vector normalization and we use the fact that the pooling aggregator can approximate any Hausdorff continuous function to an arbitrary $\epsilon$ precision \cite{qi2016pointnet}.
Note that we can always account for normalization constants (line 7 in Algorithm 1) by having aggregators prepend a unit value to all output representations; the normalization constant can then be recovered at later layers by taking the inverse of this prepended value. 
Note also that almost certainly exist settings where the symmetric functions described below can be computed exactly by the pooling aggregator (or a variant of it), but the symmetric universal approximation theorem of \cite{qi2016pointnet} along with Lipschitz continuity arguments suffice for the purposes of proving identifiability of clustering coefficients (up to an arbitrary precision).
In particular, the functions described below, that we need approximate to compute clustering coefficients, are all Lipschitz continuous on their domains (assuming we only run on nodes with positive degrees) so the errors introduced by approximation remain bounded by fixed constants (that can be made arbitrarily small). 

We assume that the weight matrices, $\mb{W}^1, \mb{W}^2$ at depths $k=2$ and $k=3$ are the identity, and that all non-linearities are rectified linear units. 
In addition, for the final iteration (i.e, $k=4$) we completely ignore neighborhood information and simply treat this layers as an MLP with a single hidden layer.
Theorem 1 can be equivalently stated as requiring $K=3$ iterations of Algorithm 1, with the representations then being fed to a single-layer MLP. 

By Lemma 3, we can assume that at depth $k=1$ all nodes in $v$'s 2-hop neighborhood have unique, one-hot indicator vectors, $\mb{h}^1_v \in \mathcal{E}_I$. 
Thus, at depth $k=2$ in Algorithm 1, suppose that we sum the unnormalized representations of the neighboring nodes. Then without loss of generality, we will have that $\mb{h}^{2}_v = \mb{h}^1_v \oplus \mb{A}_{v}$ where $\mb{A}$ is the adjacency matrix of the subgraph containing all nodes connected to $v$ in $G^4$ and $\mb{A}_v$ is the row of the adjacency matrix corresponding to $v$. 
Then, at depth $k=3$, again assume that we sum the neighboring representations (with the weight matrices as the identity), then we will have that 
\begin{equation}
\mb{h}_v^3 = \mb{h}^1_v \oplus \mb{A}_{v} \oplus \left(\sum_{v \in \mathcal{N}(v)} \mb{h}^1_v \oplus \mb{A}_{v}\right).
\end{equation}
Letting $m$ denote the dimensionality of the $\mb{h}^1_v$ vectors (i.e., $m \equiv \chi(G^4)$ from Lemma 3) and using square brackets to denote vector indexing, we can observe that
\begin{itemize}
\item
	$\mb{a} \equiv \mb{h}^3_v[0:m]$ is $v$'s one-hot indicator vector. 
\item
	$\mb{b} \equiv \mb{h}^3_v[m:2m]$ is $v$'s row in the adjacency matrix, $\mb{A}$.
\item
	$\mb{c} \equiv \mb{h}^3_v[3m:4m]$ is the sum of the adjacency rows of $v$'s neighbors. 
\end{itemize}

Thus, we have that $\mb{b}^\top\mb{c}$ is the number of edges in the subgraph containing only $v$ and it's immediate neighbors and $\sum_{i=0}^{m}\mb{b}[i] = d_v$.
Finally we can compute
\begin{align}
\frac{2(\mb{b}^\top\mb{c} - d_v)}{(d_v)(d_v - 1)}  &= \frac{2 | \{ e_{v,v'} : v, v' \in \mathcal{N}(v), e_{v,v'} \in \E \}|}{(d_v)(d_v - 1)}\\
&=c_v,
\end{align}
and since this is a continuous function of $\mb{h}^{3}_v$, we can approximate it to an arbitrary $\epsilon$ precision with a single-layer MLP (or equivalently, one more iteration of Algorithm 1, ignoring neighborhood information). 
Again this last step follows directly from \cite{hornik1991approximation}.
\end{proof}

\renewcommand{\thetheorem}{2}
\begin{corollary}
Suppose we sample nodes features from any probability distribution $\mu$ over $\mb{x} \in U$, where $\mu$ is absolutely continuous with respect to the Lebesgue measure. Then the conditions of Theorem 1 are almost surely satisfied with feature inputs $\mb{x}_v \sim \mu$. 
\end{corollary}

Corollary 2 is a direct consequence of Theorem 1 and the fact that, for any probability distribution that is absolutely continuous w.r.t. the Lebesgue measure, the probability of sampling two identical points is zero. 
Empirically, we found that \name-pool was in fact capable of maintaining modest performance by leveraging graph structure, even with completely random feature inputs (see Figure \ref{fig:noise}).
However, the performance \name-GCN was not so robust, which makes intuitive sense given that the Lemmas 1, 2, and 3 rely directly on the universal expressive capability of the pooling aggregator. 

\renewcommand{\thetheorem}{3}
\begin{figure}
\centering
\includegraphics[scale=1.0]{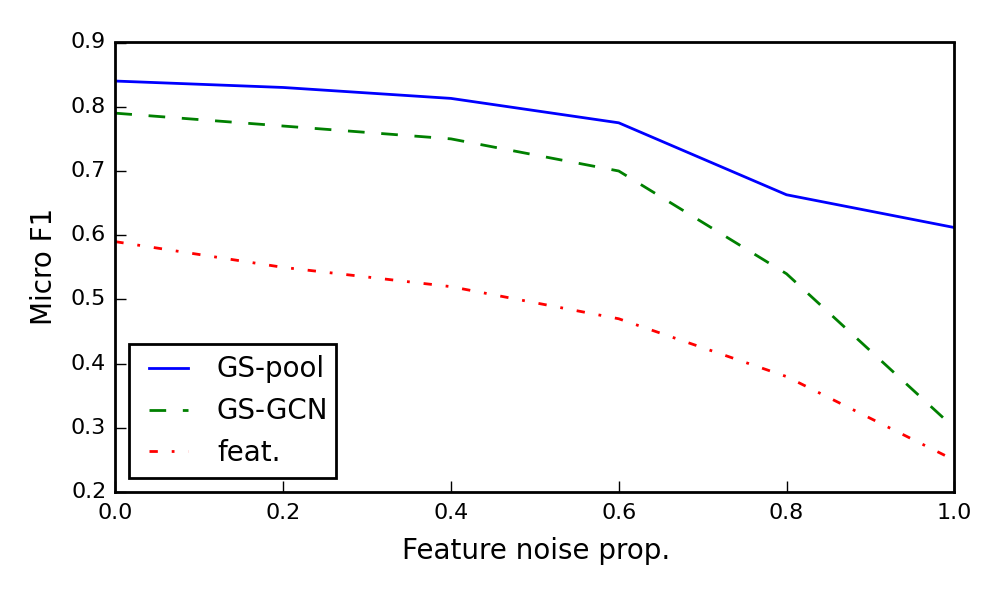}
\caption{Accuracy (in F1-score) for different approaches on the citation data as the feature matrix is incrementally replaced with random Gaussian noise.}
\label{fig:noise}
\end{figure}

Finally, we note that Theorem 1 and Corollary 2 are expressed with respect to a particular given graph and are thus somewhat transductive.
For the inductive setting, we can state
\begin{corollary}
Suppose that for all graphs $\G=(\V,\E)$ belonging to some class of graphs $G^*$, we have that $\exists k,d \geq 0, k,d \in \mathbb{Z}$ such that
$$\mb{h}^k_v \neq \mb{h}^k_{v'}, \forall (v,v') \in \{(v,v') : \exists u \in \V, v,v' \in \mathcal{N}^3(u)\}, \mb{h}^k_v, \mb{h}^k_{v'} \in \mathcal{E}^{d}_I,$$
then we can approximate clustering coefficients to an arbitrary epsilon after $K=k+4$ iterations of Algorithm \ref{alg:basic}.
\end{corollary}
Corollary 3 simply states that if after $k$ iterations of Algorithm 1, we can learn to uniquely identify nodes for a class of graphs, then we can also approximate clustering coefficients to an arbitrary precision for this class of graphs.

%% file: paper-graphnet.bbl
\begin{thebibliography}{10}

\bibitem{abadi2016tensorflow}
M.~Abadi, A.~Agarwal, P.~Barham, E.~Brevdo, Z.~Chen, C.~Citro, G.~S. Corrado,
  A.~Davis, J.~Dean, M.~Devin, et~al.
\newblock Tensorflow: Large-scale machine learning on heterogeneous distributed
  systems.
\newblock {\em arXiv preprint $ $}, 2016.

\bibitem{arora2017simple}
S.~Arora, Y.~Liang, and T.~Ma.
\newblock A simple but tough-to-beat baseline for sentence embeddings.
\newblock In {\em ICLR}, 2017.

\bibitem{benson2016higher}
A.~R. Benson, D.~F. Gleich, and J.~Leskovec.
\newblock Higher-order organization of complex networks.
\newblock {\em Science}, 353(6295):163--166, 2016.

\bibitem{bruna2013spectral}
J.~Bruna, W.~Zaremba, A.~Szlam, and Y.~LeCun.
\newblock Spectral networks and locally connected networks on graphs.
\newblock In {\em ICLR}, 2014.

\bibitem{cao2015grarep}
S.~Cao, W.~Lu, and Q.~Xu.
\newblock Grarep: Learning graph representations with global structural
  information.
\newblock In {\em KDD}, 2015.

\bibitem{chen2017stochastic}
J.~Chen and J.~Zhu.
\newblock Stochastic training of graph convolutional networks.
\newblock {\em arXiv preprint arXiv:1710.10568}, 2017.

\bibitem{dai2016discriminative}
H.~Dai, B.~Dai, and L.~Song.
\newblock Discriminative embeddings of latent variable models for structured
  data.
\newblock In {\em ICML}, 2016.

\bibitem{defferrard2016convolutional}
M.~Defferrard, X.~Bresson, and P.~Vandergheynst.
\newblock Convolutional neural networks on graphs with fast localized spectral
  filtering.
\newblock In {\em NIPS}, 2016.

\bibitem{duvenaud2015convolutional}
D.~K. Duvenaud, D.~Maclaurin, J.~Iparraguirre, R.~Bombarell, T.~Hirzel,
  A.~Aspuru-Guzik, and R.~P. Adams.
\newblock Convolutional networks on graphs for learning molecular fingerprints.
\newblock In {\em NIPS}, 2015.

\bibitem{gori2005new}
M.~Gori, G.~Monfardini, and F.~Scarselli.
\newblock A new model for learning in graph domains.
\newblock In {\em IEEE International Joint Conference on Neural Networks},
  volume~2, pages 729--734, 2005.

\bibitem{grover2016node2vec}
A.~Grover and J.~Leskovec.
\newblock node2vec: Scalable feature learning for networks.
\newblock In {\em KDD}, 2016.

\bibitem{hamilton2016diachronic}
W.~L. Hamilton, J.~Leskovec, and D.~Jurafsky.
\newblock Diachronic word embeddings reveal statistical laws of semantic
  change.
\newblock In {\em ACL}, 2016.

\bibitem{he2016identity}
K.~He, X.~Zhang, S.~Ren, and J.~Sun.
\newblock Identity mappings in deep residual networks.
\newblock In {\em EACV}, 2016.

\bibitem{hochreiter1997long}
S.~Hochreiter and J.~Schmidhuber.
\newblock Long short-term memory.
\newblock {\em Neural Computation}, 9(8):1735--1780, 1997.

\bibitem{hornik1991approximation}
K.~Hornik.
\newblock Approximation capabilities of multilayer feedforward networks.
\newblock {\em Neural Networks}, 4(2):251--257, 1991.

\bibitem{kingma2014adam}
D.~Kingma and J.~Ba.
\newblock Adam: A method for stochastic optimization.
\newblock In {\em ICLR}, 2015.

\bibitem{kipf2016semi}
T.~N. Kipf and M.~Welling.
\newblock Semi-supervised classification with graph convolutional networks.
\newblock In {\em ICLR}, 2016.

\bibitem{kipf2016variational}
T.~N. Kipf and M.~Welling.
\newblock Variational graph auto-encoders.
\newblock In {\em NIPS Workshop on Bayesian Deep Learning}, 2016.

\bibitem{kruskal1964multidimensional}
J.~B. Kruskal.
\newblock Multidimensional scaling by optimizing goodness of fit to a nonmetric
  hypothesis.
\newblock {\em Psychometrika}, 29(1):1--27, 1964.

\bibitem{levy2014neural}
O.~Levy and Y.~Goldberg.
\newblock Neural word embedding as implicit matrix factorization.
\newblock In {\em NIPS}, 2014.

\bibitem{li2015gated}
Y.~Li, D.~Tarlow, M.~Brockschmidt, and R.~Zemel.
\newblock Gated graph sequence neural networks.
\newblock In {\em ICLR}, 2015.

\bibitem{mikolov2013distributed}
T.~Mikolov, I.~Sutskever, K.~Chen, G.~S. Corrado, and J.~Dean.
\newblock Distributed representations of words and phrases and their
  compositionality.
\newblock In {\em NIPS}, 2013.

\bibitem{ng2001spectral}
A.~Y. Ng, M.~I. Jordan, Y.~Weiss, et~al.
\newblock On spectral clustering: Analysis and an algorithm.
\newblock In {\em NIPS}, 2001.

\bibitem{niepert2016learning}
M.~Niepert, M.~Ahmed, and K.~Kutzkov.
\newblock Learning convolutional neural networks for graphs.
\newblock In {\em ICML}, 2016.

\bibitem{page1999pagerank}
L.~Page, S.~Brin, R.~Motwani, and T.~Winograd.
\newblock The pagerank citation ranking: Bringing order to the web.
\newblock Technical report, Stanford InfoLab, 1999.

\bibitem{scikit-learn}
F.~Pedregosa, G.~Varoquaux, A.~Gramfort, V.~Michel, B.~Thirion, O.~Grisel,
  M.~Blondel, P.~Prettenhofer, R.~Weiss, V.~Dubourg, J.~Vanderplas, A.~Passos,
  D.~Cournapeau, M.~Brucher, M.~Perrot, and E.~Duchesnay.
\newblock Scikit-learn: Machine learning in {P}ython.
\newblock {\em Journal of Machine Learning Research}, 12:2825--2830, 2011.

\bibitem{pennington2014glove}
J.~Pennington, R.~Socher, and C.~D. Manning.
\newblock Glove: Global vectors for word representation.
\newblock In {\em EMNLP}, 2014.

\bibitem{perozzi2014deepwalk}
B.~Perozzi, R.~Al-Rfou, and S.~Skiena.
\newblock Deepwalk: Online learning of social representations.
\newblock In {\em KDD}, 2014.

\bibitem{qi2016pointnet}
C.~R. Qi, H.~Su, K.~Mo, and L.~J. Guibas.
\newblock Pointnet: Deep learning on point sets for 3d classification and
  segmentation.
\newblock In {\em CVPR}, 2017.

\bibitem{rehurek_lrec}
R.~{\v R}eh{\r u}{\v r}ek and P.~Sojka.
\newblock {Software Framework for Topic Modelling with Large Corpora}.
\newblock In {\em LREC}, 2010.

\bibitem{scarselli2009graph}
F.~Scarselli, M.~Gori, A.~C. Tsoi, M.~Hagenbuchner, and G.~Monfardini.
\newblock The graph neural network model.
\newblock {\em IEEE Transactions on Neural Networks}, 20(1):61--80, 2009.

\bibitem{shervashidze2011weisfeiler}
N.~Shervashidze, P.~Schweitzer, E.~J.~v. Leeuwen, K.~Mehlhorn, and K.~M.
  Borgwardt.
\newblock Weisfeiler-lehman graph kernels.
\newblock {\em Journal of Machine Learning Research}, 12:2539--2561, 2011.

\bibitem{siegal1956nonparametric}
S.~Siegal.
\newblock {\em Nonparametric statistics for the behavioral sciences}.
\newblock McGraw-hill, 1956.

\bibitem{subramanian2005gene}
A.~Subramanian, P.~Tamayo, V.~K. Mootha, S.~Mukherjee, B.~L. Ebert, M.~A.
  Gillette, A.~Paulovich, S.~L. Pomeroy, T.~R. Golub, E.~S. Lander, et~al.
\newblock Gene set enrichment analysis: a knowledge-based approach for
  interpreting genome-wide expression profiles.
\newblock {\em Proceedings of the National Academy of Sciences},
  102(43):15545--15550, 2005.

\bibitem{tang2015line}
J.~Tang, M.~Qu, M.~Wang, M.~Zhang, J.~Yan, and Q.~Mei.
\newblock Line: Large-scale information network embedding.
\newblock In {\em WWW}, 2015.

\bibitem{wang2016structural}
D.~Wang, P.~Cui, and W.~Zhu.
\newblock Structural deep network embedding.
\newblock In {\em KDD}, 2016.

\bibitem{wang2017community}
X.~Wang, P.~Cui, J.~Wang, J.~Pei, W.~Zhu, and S.~Yang.
\newblock Community preserving network embedding.
\newblock In {\em AAAI}, 2017.

\bibitem{watts1998collective}
D.~J. Watts and S.~H. Strogatz.
\newblock Collective dynamics of ‘small-world’ networks.
\newblock {\em Nature}, 393(6684):440--442, 1998.

\bibitem{xu2017embedding}
L.~Xu, X.~Wei, J.~Cao, and P.~S. Yu.
\newblock Embedding identity and interest for social networks.
\newblock In {\em WWW}, 2017.

\bibitem{yang2016revisiting}
Z.~Yang, W.~Cohen, and R.~Salakhutdinov.
\newblock Revisiting semi-supervised learning with graph embeddings.
\newblock In {\em ICML}, 2016.

\bibitem{zitnik2017tissue}
M.~Zitnik and J.~Leskovec.
\newblock Predicting multicellular function through multi-layer tissue
  networks.
\newblock {\em Bioinformatics}, 33(14):190--198, 2017.

\end{thebibliography}
